\documentclass[12pt,reqno]{amsart}

\usepackage{amsmath,amssymb,amsthm,bm,mathrsfs,braket,color}
\usepackage{graphicx}
\makeatletter
\@addtoreset{equation}{section}
\makeatother

\theoremstyle{plain}
\newtheorem{theorem}{Theorem}[section]
\newtheorem{proposition}[theorem]{Proposition}
\newtheorem{lemma}[theorem]{Lemma}

\newtheorem{corollary}[theorem]{Corollary}

\theoremstyle{definition}
\newtheorem{example}[theorem]{Example}
\theoremstyle{remark}
\newtheorem*{remark}{Remark}

\newcommand{\bsq}{\blacksquare}

\providecommand{\kak}[1]{(\ref{#1})}
\providecommand{\add}{a^*}
\newcommand{\til}{\bar}
\providecommand{\eq}[1]{\begin{equation}\label{#1}}
\providecommand{\en}{\end{equation}}
\providecommand{\con}{{\mathscr C}^+}
\providecommand{\conn}{{\mathscr C}^+_0}
\providecommand{\RR}{\mathbb R}
\providecommand{\pp}{{\rm p}}
\providecommand{\qq}{{\sigma}}
\providecommand{\OO}{{\rm S}}
\providecommand{\LR}{L^2({\mathbb R})}
\providecommand{\vvv}[1]{\left(\!\!\!\begin{array}{c} #1\end{array}\!\!\!\right)}
\providecommand{\hhh}{{\mathcal H}}

\newcommand{\cH}{\mathcal{H}}
\newcommand{\BR}{\mathbb{R}}
\newcommand{\BC}{\mathbb{C}}
\newcommand{\Ran}{\mathop{\mathrm{Ran}}}
\newcommand{\BN}{\mathbb{N}}
\def\norm#1{\| {#1} \|}
\newcommand{\inner}[2]{ \left(#1,#2\right) }
\newcommand{\diag}{\mathrm{diag}}
\newcommand{\ep}{\epsilon}
\newcommand{\UU}{\widetilde{\Phi}}

\title[spectral analysis of NcHO]{
Spectral analysis of non-commutative harmonic oscillators: 
the lowest eigenvalue and no crossing}
\author{Fumio Hiroshima}
\address{Faculty of Mathematics, Kyushu University,  Fukuoka, 819-0395, Japan}
\email{hiroshima@math.kyushu-u.ac.jp}
\author{Itaru Sasaki}
\address{ Department of Mathematical Sciences,
 Shinshu University, Matsumoto 390--8621, Japan}
\email{isasaki@shinshu-u.ac.jp}

\date{\today}


\begin{document}
%%%%%%%%%%%%%%%%
\begin{abstract}
The lowest eigenvalue of 
non-commutative harmonic oscillators 
$Q(\alpha,\beta)$ ($\alpha>0,\beta>0, \alpha\beta>1$)  is studied.
It is shown that  $Q(\alpha,\beta)$ can be 
decomposed into four self-adjoint operators, 
$$\displaystyle Q(\alpha,\beta)=\bigoplus_{\qq=\pm, \pp=1,2} Q_{\qq \pp},$$
 and all the  eigenvalues  of each operator $Q_{\qq \pp}$ are simple.
We show  that 
the lowest eigenvalue  of $Q(\alpha,\beta)$ 
is simple whenever 
$\alpha\neq \beta$.
Furthermore a Jacobi matrix representation of $Q_{\qq\pp}$ is given and 
spectrum  of $Q_{\qq\pp}$ is considered numerically. 
\end{abstract}
\maketitle
%%%%%%%%%%%%%%%%%%%%%%%%%%%%%%%%%%%%%%%%%%%%

%\setlength{\baselineskip}{15pt}

\section{Introduction}
The non-commutative harmonic oscillator is introduced by 
A.~Parmeggiani and M.~Wakayama \cite{MR1811870,pw2002,pw2003} 
as a non-commutative extension of harmonic oscillators. 
We also refer to  \cite{p2010} which 
is a first account about non-commutative harmonic oscillators  and of their spectral properties. 
It is defined by 
\begin{align}
 Q = Q (\alpha,\beta) = A\otimes 
  \left(-\frac{1}{2}\frac{d^2}{dx^2}+\frac{1}{2}x^2 \right)
      + J \otimes 
      \left(x\frac{d}{dx}+\frac{1}{2}\right),
\end{align}
as 
an operator in $\hhh={\mathbb C}^2\otimes L^2(\RR)$. 
Here $A, J\in {\rm Mat}_2(\RR)$, $A$ is positive definite symmetric, and $J$ skew-symmetric. 
Furthermore $A+iJ$ is positive definite. 
It is shown in \cite{pw2002,pw2003} that $A$ and $J$ can be  assumed to be 
 $A= \begin{pmatrix} \alpha & 0\\ 0&\beta\end{pmatrix}$,
$J= \begin{pmatrix}0 & -1 \\ 1&0\end{pmatrix}$, 
and  $\alpha$ and $\beta$  satisfy 
 \eq{ast}
\alpha>0,\quad \beta>0,\quad   \alpha\beta>1.
\en
We fix $A$ and $J$ as above, and 
throughout  this paper we assume 
\kak{ast}.
Under \kak{ast}, $Q$ is self-adjoint on the domain 
 $D(Q)={\mathbb C}^2\otimes (
 D(d^2/dx^2)\cap D(x^2))$ and 
 has purely discrete spectrum
$E_0\leq E_1 \leq E_2 \leq \cdots \nearrow \infty$.
When $\alpha=\beta$, $Q(\alpha,\beta)$ is equivalent to the direct sum of a harmonic oscillator.
Then $E_j=E_{j+1}$ for $j=0,2,4,\cdots$. 
On the other hand, when 
$\alpha\neq\beta$,  
$Q(\alpha,\beta)$  is regarded as a $q$-deformation of harmonic oscillators with $q=\beta/\alpha$ or $q=\alpha/\beta$, 
and the spectral analysis of  $Q$ is nontrivial. 

An  eigenvector associated with the lowest eigenvalue $E=E_0$ 
is called a ground state in this paper.
 A long-standing problem concerning  eigenvalues of $Q(\alpha,\beta)$ is to determine 
their  multiplicity explicitly. 
Let $\alpha\ne\beta$. 
Let $E_n=E_n(\alpha,\beta)$ denote the $n$-th eigenvalue of $Q(\alpha,\beta)$. 
The map 
$c_n:(\alpha,\beta)\mapsto E_n(\alpha,\beta)\in\RR$ is called 
an  eigenvalue-curve.
To consider the multiplicity of eigenvalues is reduced to 
studying  crossing or no crossing of eigenvalue-curves.  

We state a short history concerning 
studies of the multiplicity of eigenvalues of $Q$. 
In \cite{pw2003} 
it is shown that the multiplicity of any eigenvalues 
of $Q$ is at most three and an alternative proof is given 
in 
 \cite{ochi2001}.
In a numerical  level 
it is  found in \cite{nnw2002}  that 
eigenvalue-curves cross at some points 
but  the lowest eigenvalue is simple. 
The multiplicity of eigenvalues of $Q$  is also considered 
in \cite{iw07}, where it is derived that  
$$\left(n-\frac{1}{2}\right)\min\{\alpha,\beta\} \sqrt{\frac{\alpha\beta-1}{\alpha\beta}}\leq E_{2n-1}\leq E_{2n}\leq \left(n-\frac{1}{2}\right)\max\{\alpha,\beta\}\sqrt{\frac{\alpha\beta-1}{\alpha\beta}}$$
for $n=1,2,3,\cdots$. From this we can see that the multiplicity of $E$ is at most two if $\beta<3\alpha$ or $\alpha<3\beta$. 
In \cite{p2004} it is shown that $E$ is simple but for sufficiently large $\alpha\beta$. 
Furthermore in  \cite{1207.4060} 
it is proven that the lowest eigenvalue is at most two and all the ground state are even 
for 
$(\alpha, \beta)\in D_{\sqrt 2}$, where 
$D_{\sqrt 2}=\{(\alpha,\beta)| \alpha, \beta>\sqrt 2\}$, 
and it is also shown 
that  $E$ is simple   for 
$(\alpha, \beta)\in D$ with some subset  $D\subset D_{\sqrt 2}$.
Recently Wakayama  \cite{wakayama:2012}  breaks through in studying 
the multiplicity of $E$. 
It 
is proven that 
if  all the ground states are even, 
then $E$  is simple  whenever  $\alpha\neq \beta$.
Combining \cite{wakayama:2012} with  \cite{1207.4060}, it is immediate to see that 
$E$  is simple for $(\alpha,\beta)\in D_{\sqrt 2}$.

In this paper we settle down the question concerning the multiplicity of the lowest eigenvalue of $Q$, i.e., we prove that $E$ is simple for all 
values of  $\alpha$ and $\beta$ ($\alpha\ne\beta$) in Theorem \ref{uniqthm}. 
Moreover  
no crossing between  eigenvalue-curves associated 
with an odd eigenvector and 
an even eigenvector is also proven in Corollary \ref{nocrossing}.

This paper is organized as follows. 
In Section \ref{sect.decom}, we decompose $Q(\alpha,\beta)$
into four self-adjoint operators: 
  $Q(\alpha,\beta)=\bigoplus_{\qq=\pm, \pp=1,2} Q_{\qq \pp}$. 
  It is shown that each $Q_{\qq\pp}$ is  equivalent to 
some Jacobi matrix  $\widehat{Q}_{\qq\pp}$, and 
 all the eigenvalues of $Q_{\qq\pp}$ are simple.
In Section \ref{sec.simplicity}, we show that the lowest eigenvalue of $Q(\alpha,\beta)$
is simple. 
In Section \ref{sec.pio}, we construct  a unitary transformation
 $Q_{\qq\pp}\mapsto \bar{Q}_{\qq\pp}$ such that $e^{-t\bar{Q}_{\qq\pp}}$
is positivity improving, and it is shown  that the ground state is in a positive cone. 
In Section \ref{sec.difference}, 
we show that 
$
 \widehat{Q}_{-\pp} - \widehat{Q}_{+\pp} \geq \Delta(\alpha,\beta)$, $\pp=1,2$,  
 with some 
$\Delta(\alpha,\beta)$. 
In particular, if $\Delta(\alpha,\beta)>0$, then 
 there  is no  crossing between 
the $n$-th eigenvalue-curve of $Q_{-\pp}$ and that of $Q_{+\pp}$.  
In  Section \ref{sec.finite}, we show some numerical results.

\section{Decomposition of $Q(\alpha,\beta)$ and Jacobi matrix}{\label{sect.decom}}
\subsection{Decomposition of $Q(\alpha,\beta)$}
Let $a=\frac{1}{\sqrt{2}}(x+\frac{d}{dx})$ and 
 $a^*=\frac{1}{\sqrt{2}}(x-\frac{d}{dx})$
 be the annihilation operator and the creation operators, respectively.
 In terms of $a$ and $\add$, $Q$ can be expressed as 
\begin{align}
  Q=A(a^* a + \frac{1}{2})+ \frac{J}{2}(aa-\add \add).
\end{align}
Let $\cH_+$ (resp. $\cH_-$)  be the set of even (resp. odd) 
functions in $\cH$, and 
 $P_+$ (resp. $P_-$) be the orthogonal projection onto $\cH_+$ 
(resp.  $\cH_-$).
Let $\ket{n}$ be the  $n$-th normalized eigenvector of $a^* a$, i.e.,
$
  \ket{n} = \frac{1}{\sqrt{n!}} (a^*)^n \ket{0}$
with $\ket{0} = \pi^{-1/4}e^{-x^2/2}$. 
Let ${\mathbb C}\ket {n}$ be the one-dimensional subspace spanned by $\ket{n}$ over ${\mathbb C}$. 
Hence the Wiener-It\^o decomposition $\LR=\bigoplus_{n=0}^\infty {\mathbb C}\ket{n}$ follows. 
The total Hilbert space is 
$$\hhh
\cong
\left\{\left. \vvv{X\\Y}\, \right|\, X,Y\in \bigoplus_{n=0}^\infty {\mathbb C} \ket{n}\right\}
\cong 
\bigoplus_{n=0}^\infty  \hhh_n,\quad 
\hhh_n=\vvv{{\mathbb C}\ket{n}\\{\mathbb C}\ket{n}}.$$
We use this equivalence without noticing. 
Since $a\ket{n}= \sqrt n \ket{n-1}$ and  
$a^*\ket{n}=\sqrt{n+1}\ket{n+1}$, 
we see that 
$aa:\hhh_n\to\hhh_{n-2}$ and 
$\add \add :\hhh_n \to\hhh_{n+2}$. 
Furthermore 
$a^\ast a$ leaves $\hhh_n$  invariant. 
Then we have $Q:\hhh_n\to \hhh_{n-2}\oplus\hhh_n\oplus\hhh_{n+2}$. 
From these observation we can find invariant domains of $Q$. 
We denote  the orthogonal projection 
onto ${\mathbb C}\ket{n}$
 by 
$\ket{n}\!\bra{n}$, and 
define  orthogonal projections on $\hhh$ by
\begin{align}
  P_\uparrow(n) = 
  \begin{pmatrix} \ket{n}\! \bra{n}& 0 \\ 0 &0 \end{pmatrix}, 
  \qquad 
  P_\downarrow(n) = 
  \begin{pmatrix} 0 &0 \\  0 & \ket{n}\! \bra{n}  
  \end{pmatrix}. 
\end{align}
Note that 
$  1 = \sum_{n=0}^\infty (P_\uparrow(n) + P_\downarrow(n))$.
In order to decompose $Q$, we define the following orthogonal projections:
\begin{align*}
%\begin{array}{ll}
 T_{+1} &= \sum_{n=0}^\infty (P_\uparrow(4n) + P_\downarrow(4n+2) ),&
  T_{+2} &= \sum_{n=0}^\infty (P_\downarrow(4n) + P_\uparrow(4n+2)  ),\\
 T_{-1} &= \sum_{n=0}^\infty (P_\uparrow(4n+1) + P_\downarrow(4n+3) ),& 
  T_{-2} &= \sum_{n=0}^\infty (P_\downarrow(4n+1) + P_\uparrow(4n+3) ).
%\end{array}
\end{align*}
Since $\ket{2n}$ is even and $\ket{2n+1}$ is odd, one has
$T_{+1} + T_{+2} = P_+$ and $T_{-1} + T_{-2} = P_-$.
We set $\hhh_{\qq\pp}=\Ran(T_{\qq\pp})$. 
Then $\hhh$ is decomposed as 
\begin{align}
\hhh=\bigoplus_{\qq=\pm,\pp=1,2}\hhh_{\qq\pp}.
\end{align}
 \begin{theorem}
Operator $Q$ is reduced by  
 $\hhh_{\qq\pp}$, $\qq=\pm$, $\pp=1,2$.  
 \end{theorem}
%%%%%%%%%%%%%%%%
 \begin{proof}
We see that 
$ a^2 P_j(n) \supset P_j(n-2) a^2$, 
$
\add \add
 P_j(n) \supset P_j(n+2) \add \add$ and 
 $
   a^* a P_j(n) \supset P_j(n) a^* a $ 
   for all $n=0,1,2,\cdots$, 
    and 
$j=\uparrow,\downarrow$.
Clearly
it holds that 
$  A P_j(n) = P_j(n) A$, 
$
  J P_\uparrow(n) = P_\downarrow(n) J$ and 
  $
  J P_\downarrow(n) = P_\uparrow(n) J$.
Then 
$ Q T_{\qq\pp} \supset  T_{\qq\pp} Q$  
and  the theorem  follows. 
  \end{proof}
Let 
us set    $Q_{\qq\pp}=Q\lceil_{\hhh_{\qq\pp}}$.
Then it holds that 
%$Q_+ = Q_{+1} \oplus Q_{+2}$, $Q_- = Q_{-1} \oplus Q_{-2}$ and  
 \begin{align}
Q=\bigoplus_{\qq=\pm,\pp=1,2}Q_{\qq\pp}.
\end{align}

{\tiny
%%%%%%%%%%%%%%%%
 \begin{figure}[ht]
\begin{tabular}[tb]{c|c|c|c|c|c|c|c|c|c|c|c|c|c|c|} \hline
            $n$      & 0&1&  2&3&4&5&6&7&8&9&10&11&12& $\cdots$  \\ \hline
$\uparrow$ &$\bsq$&$\Box$&$\Box$&$\Box$&$\bsq$&$\Box$&$\Box$&$\Box$&$\bsq$&$\Box$&$\Box$&$\Box$&$\bsq$& $\cdots$ \\ \hline
$\downarrow$&$\Box$&$\Box$& $\bsq$& $\Box$&$\Box$&$\Box$&$\bsq$&$\Box$&$\Box$&$\Box$&$\bsq$&$\Box$&$\Box$& $\cdots$ \\ \hline
\end{tabular}
\caption{$\Ran T_{+1}$ is supported on ``$\bsq$''}
%%%%%%%%%%%%%%%%
\begin{tabular}[tb]{c|c|c|c|c|c|c|c|c|c|c|c|c|c|c|} \hline
          $n$        & 0&1&  2&3&4&5&6&7&8&9&10&11&12& $\cdots$  \\ \hline
$\uparrow$ &$\Box$&$\Box$&$\bsq$&$\Box$&$\Box$&$\Box$&$\bsq$&$\Box$&$\Box$&$\Box$&$\bsq$&$\Box$&$\Box$& $\cdots$ \\ \hline
$\downarrow$&$\bsq$&$\Box$& $\Box$& $\Box$&$\bsq$&$\Box$&$\Box$&$\Box$&$\bsq$&$\Box$&$\Box$&$\Box$&$\bsq$& $\cdots$ \\ \hline
\end{tabular}
\caption{$\Ran T_{+2}$ is supported on ``$\bsq$''}
%%%%%%%%%%%%%%%%
\begin{tabular}[tb]{c|c|c|c|c|c|c|c|c|c|c|c|c|c|c|} \hline
      $n$              & 0&1&  2&3&4&5&6&7&8&9&10&11&12& $\cdots$  \\ \hline
$\uparrow$ &$\Box$&$\bsq$&$\Box$&$\Box$&$\Box$&$\bsq$&$\Box$&$\Box$&$\Box$&$\bsq$&$\Box$&$\Box$&$\Box$& $\cdots$ \\ \hline
$\downarrow$&$\Box$&$\Box$& $\Box$& $\bsq$&$\Box$&$\Box$&$\Box$&$\bsq$&$\Box$&$\Box$&$\Box$&$\bsq$&$\Box$& $\cdots$ \\ \hline
\end{tabular}
\caption{$\Ran T_{-1}$ is supported on ``$\bsq$''}
%%%%%%%%%%%%%%%%
\begin{tabular}[tb]{c|c|c|c|c|c|c|c|c|c|c|c|c|c|c|} \hline
  $n$                  & 0&1&  2&3&4&5&6&7&8&9&10&11&12& $\cdots$  \\ \hline
$\uparrow$ &$\Box$&$\Box$&$\Box$&$\bsq$&$\Box$&$\Box$&$\Box$&$\bsq$&$\Box$&$\Box$&$\Box$&$\bsq$&$\Box$& $\cdots$ \\ \hline
$\downarrow$&$\Box$&$\bsq$& $\Box$& $\Box$&$\Box$&$\bsq$&$\Box$&$\Box$&$\Box$&$\bsq$&$\Box$&$\Box$&$\Box$& $\cdots$ \\ \hline
\end{tabular}
\caption{$\Ran T_{-2}$ is supported on ``$\bsq$''}
 \end{figure}
}

\subsection{Jacobi matrix representation of $Q_{\qq\pp}$}{\label{sec.jacobi}}
%In this section, we prove that each $Q_{\qq\pp}$ is unitarily equivalent to some Jacobi matrix.
%Before giving the Jacobi matrix representation, we construct unitary transformations, which make
%the semigroup of transformed $Q_{\qq\pp}$ are positivity improving.
We construct a unitary operator implementing equivalence between $Q_{\qq\pp}$ and a Jacobi matrix. 
Set
\begin{align}
 U_{+1} = \sum_{n=0}^\infty (P_\uparrow(8n) +  P_\downarrow(8n+2))
                        - \sum_{n=0}^\infty (P_\uparrow(8n+4) + P_\downarrow(8n+6)).
\end{align}
This operator  is unitary on $\hhh_{+1}$ and we have 
\begin{align}
 \til{Q}_{+1}
= U_{+1}^{-1} Q_{+1} U_{+1} 
= T_{+1} \left(  A (a^* a + \frac{1}{2}) - \frac{\OO }{2}(aa+\add\add) \right)  T_{+1},
\end{align}
where
$
  \OO  = \begin{pmatrix}      0&1 \\ 1 & 0     \end{pmatrix}$.
In a similar way to $U_{+1}$  one can define the 
unitary operators 
$U_{+2}, U_{-1}$ and $ U_{-2}$
on $\hhh_{+2}$, $\hhh_{-1}$ and  $\hhh_{-2}$, respectively, such that
\begin{align*}
 \til{Q}_{+2}
= U_{+2}^{-1} Q_{+1} U_{+2} 
= T_{+2} \left(  A (a^* a + \frac{1}{2}) - \frac{\OO }{2}(aa+\add\add) \right)  T_{+2},\\
 \til{Q}_{-1}
= U_{-1}^{-1} Q_{-1} U_{-1} 
= T_{-1} \left(  A (a^* a + \frac{1}{2}) - \frac{\OO }{2}(aa+\add\add) \right)  T_{-1},\\
 \til{Q}_{-2}
= U_{-2}^{-1} Q_{-2} U_{-2} 
= T_{-2} \left(  A (a^* a + \frac{1}{2}) - \frac{\OO }{2}(aa+\add\add) \right)  T_{-2}.
\end{align*}
For sequences $a=(a_0, a_1, a_2,\cdots)$ and $b=(b_0,b_1,,b_2,\cdots)$, 
we define the Jacobi matrix
\begin{align}
 J(a,b) = 
\begin{pmatrix}
  b_0 & a_0 &       &            & \text{\Large 0}\\
  a_0 & b_1 & a_1 & \\
        &  a_1 & b_2 &  \ddots & \\
        &        &  \ddots & \ddots & \ddots \\
\text{\Large 0}       & &        &  \ddots & \ddots
\end{pmatrix},
\end{align}
which acts in the set of square summable sequences, $\ell^2:=\ell^2(\mathbb{N}_0)$, where 
$\mathbb{N}_0 = \mathbb{N}\cup\{0\}$.
Set $a_\qq=(a_\qq(0),a_\qq(1),\cdots)$ and $b_{\qq\pp}=(b_{\qq\pp}(0), b_{\qq\pp}(1),\cdots)$, where 
\begin{align*}
& a_+(n) = -\sqrt{(2n+1)(2n+2)}, \qquad a_-(n) = -\sqrt{(2n+2)(2n+3)},  \\
& b_{+1}(n) = \begin{cases} 
	                \alpha(1+4n) \qquad  \text{for even } n\\
	                 \beta(1+4n)  \qquad  \text{for odd } n,
	             \end{cases} \quad
	              b_{+2}(n) = b_{+1}(n)\Big|_{(\alpha,\beta)\to(\beta,\alpha)}, \\
& b_{-1}(n) = \begin{cases} 
	                \alpha(3+4n) \qquad  \text{for even } n\\
	                \beta(3+4n)  \qquad  \text{for odd } n,
	      \end{cases}\quad 
	       b_{-2}(n) = b_{-1}(n)\Big|_{(\alpha,\beta)\to(\beta,\alpha)}.
\end{align*}
For $\qq=\pm$ and $\pp=1,2$, we define  the Jacobi matrix $\widehat Q_{\qq\pp}$ by 
\begin{align}
 \widehat{Q}_{\qq\pp} &= \frac{1}{2} J( a_\qq, b_{\qq\pp}).
 \end{align}
Set $\left\{
\vvv{\ket{4n}\\0}, \vvv{0\\ \ket{4n+2}},n=0,1,2,....\right\}$ 
is a complete orthonormal system of 
$\hhh_{+1}$. 
Let $e_n=(\delta_{n,j})_{j=0}^\infty \in \ell^2$ be the standard basis of $\ell^2$. 
We define the unitary operator $Y_{+1}:\hhh_{+1} \to \ell^2$
by 
$ Y_{+1} \vvv{\ket{4n}\\0} = e_{2n}$ and $ Y_{+1} \vvv{0\\ \ket{4n+2} }= e_{2n+1}$.
Then one can compute the matrix element of $\til{Q}_{+1}$ as $\widehat{Q}_{+1}=Y_{+1}\bar{Q}_{+1}Y_{+1}^{-1}$.
Similarly one can define the unitary transformations such that the following theorem holds.
 \begin{theorem}[Jacobi matrix representations]
For $\qq=\pm, \pp=1,2$, the operators $Q_{\qq\pp}$ are unitarily equivalent to the Jacobi matrix $\widehat{Q}_{\qq\pp}$. 
 \end{theorem}
 \begin{remark}
In the case of $\alpha=\beta$, $\widehat Q_{\qq1}=\widehat Q_{\qq2}$ for $\qq=\pm$.
Explicitly each $\widehat Q_{\qq\pp}$ is expressed as 
\begin{align*}
\begin{array}{l}
\widehat{Q}_{+1} = \frac{1}{2}
{\tiny
\begin{pmatrix}
\alpha & -\sqrt{1\!\cdot\! 2} &  &  &  &  &   \text{\huge 0}  \\ 
-\sqrt{1\!\cdot\! 2} & 5\beta & -\sqrt{3\!\cdot\! 4} \\
 & -\sqrt{3\!\cdot\! 4} & 9\alpha & -\sqrt{5\!\cdot\! 6} & & &  \\
 &  & -\sqrt{5\!\cdot\! 6} & 13\beta & -\sqrt{7\!\cdot\! 8} \\
 &  & & -\sqrt{7\!\cdot\! 8} & 17\alpha & -\sqrt{9\!\cdot\! 10}  \\
 &  &  &  & -\sqrt{9\!\cdot\! 10} & 21\beta & \ddots \\
 \text{\huge 0} &&&&& \ddots& \ddots 
\end{pmatrix},
}\\
\widehat{Q}_{+2} =\frac{1}{2}
{\tiny 
\begin{pmatrix}
\beta & -\sqrt{1\!\cdot\! 2} &  &  &  &  &   \text{\huge 0}\\
-\sqrt{1\!\cdot\! 2} & 5\alpha & -\sqrt{3\!\cdot\! 4} &  &  &  &  \\
 & -\sqrt{3\!\cdot\! 4} & 9\beta & -\sqrt{5\!\cdot\! 6} &  &  &   \\
 &  & -\sqrt{5\!\cdot\! 6} & 13\alpha & -\sqrt{7\!\cdot\! 8} &  &   \\
 &  &  & -\sqrt{7\!\cdot\! 8} & 17\beta & -\sqrt{9\!\cdot\! 10} &   \\
 &  &  &  & -\sqrt{9\!\cdot\! 10} & 21\alpha & \ddots \\
\text{\huge 0}  & & & & &\ddots & \ddots
\end{pmatrix},
}\\
\widehat{Q}_{-1} =\frac{1}{2}
{\tiny 
\begin{pmatrix}
3\alpha & -\sqrt{2\!\cdot\! 3} &  &  &  &  &    \text{\huge 0} \\
-\sqrt{2\!\cdot\! 3} & 7\beta & -\sqrt{4\!\cdot\! 5} &  &    &  &  \\
 & -\sqrt{4\!\cdot\! 5} & 11\alpha & -\sqrt{6\!\cdot\! 7}   &  &  &  \\
 &  & -\sqrt{6\!\cdot\! 7} & 15\beta & -\sqrt{8\!\cdot\! 9}   &  &  \\
 &  &  & -\sqrt{8\!\cdot\! 9} & 19\alpha & -\sqrt{10\!\cdot\! 11}  &  \\[-0.2cm]
 &  &  &  & -\sqrt{10\!\cdot\! 11} & 23\beta & \ddots \\[-0.08cm]
\text{\huge 0}  & &&&& \ddots& \ddots
\end{pmatrix}
},\\
\widehat{Q}_{-2} = \frac{1}{2}
{\tiny 
\begin{pmatrix}
3\beta & -\sqrt{2\!\cdot\! 3} &  &  &  &  &   \text{\huge 0} \\
-\sqrt{2\!\cdot\! 3} & 7\alpha & -\sqrt{4\!\cdot\! 5}   &  &  &  & \\
 & -\sqrt{4\!\cdot\! 5} & 11\beta & -\sqrt{6\!\cdot\! 7} &    &  & \\
 &  & -\sqrt{6\!\cdot\! 7} & 15\alpha & -\sqrt{8\!\cdot\! 9} &    & \\
 &  &  & -\sqrt{8\!\cdot\! 9} & 19\beta & -\sqrt{10\!\cdot\! 11}   & \\[-0.2cm]
 &  &  &  & -\sqrt{10\!\cdot\! 11} & 23\alpha & \ddots \\[-0.08cm]
\text{\huge 0}  & &&&&\ddots & \ddots
\end{pmatrix}
}.
\end{array}
\end{align*}
\end{remark}
%Now we can state the second main theorem.
\begin{theorem}
\label{42}
Each  eigenvalue of $Q_{\qq\pp}$, $\qq=\pm$,  $\pp=1,2$,  is simple.
%In particular the lowest eigenvalue of $Q_{\qq\pp}$ is simple. 
\end{theorem}
\begin{proof}
 Let $\lambda$ be any eigenvalue of $\widehat{Q}_{+1}$ with an eigenvector $u = (u_n)_{n=0}^\infty\in \ell^2$.
Then $\lambda$ and $u$ satisfy the recurrence relations:
\begin{align}
& u_{n+1} = a_+ (n)^{-1}\left(  ( \lambda-b_{+1}(n) ) u_n - a_+(n-1) u_{n-1}\right), \quad n \in \BN_0, \label{7.1rec} \\
& u_{-1}=0 . \label{7.2rec}
\end{align}
Note that $a_+(n)\neq 0$.
Solutions of system \eqref{7.1rec}-\eqref{7.2rec} are  uniquely determined  by the term $u_0 \in\BC$.
Hence the multiplicity of any eigenvalue  of $\widehat{Q}_{+1}$ is simple.
Proofs for other cases are similar.
\end{proof}
%%%%%%%%%%%%%%%%
Let $\lambda_{\qq\pp}(n)=\lambda_{\qq\pp}(n,\alpha,\beta)$ be the $n$-th eigenvector of $Q_{\qq\pp}$. 
Then $\{\lambda_{\qq\pp}(n)\}_{n=0}^\infty = Spec(Q_{\qq\pp})$ 
and $\lambda_{\qq\pp}(n) \leq \lambda_{\qq\pp}(n+1)$ for $n=0,1,2,\cdots$.
The following result follows immediately from the above theorem.
 \begin{corollary}
For each $\qq=\pm$ and $\pp=1,2$, eigenvalue-curves 
$$\{ (\alpha,\beta)\mapsto \lambda_{\qq\pp}(n)=\lambda_{\qq\pp}(n,\alpha,\beta), n=0,1,2,3,\cdots\}$$
have no crossing, i.e., for arbitrary $(\alpha,\beta)$ and $n\neq m$, 
$\lambda_{\qq\pp}(n,\alpha,\beta) \neq \lambda_{\qq\pp}(m,\alpha,\beta)$.
 \end{corollary}

\section{Simplicity of the lowest eigenvalue of  $Q(\alpha,\beta)$}{\label{sec.simplicity}}
In this section, we state  the main theorem in this paper.
%%%%%%%%%%%%%%%%
\begin{theorem}{\label{uniqthm}}
Assume that $\alpha\neq \beta$.
Then the lowest eigenvalue  of $Q(\alpha,\beta)$
is simple  and the ground state is even.
% for all $\alpha,\beta>0$ with $\alpha\beta>1$.
\end{theorem}
%%%%%%%%%%%%%%%%
In order to show 
Theorem \ref{uniqthm} we introduce a remarkable result given by 
Wakayama  \cite{wakayama:2012}. 
 \begin{theorem}{\label{waka}}
Assume that
(1) 
 $\alpha\neq \beta$;
(2) 
 all the ground states of $Q(\alpha,\beta)$ are even, i.e., $\ker(Q(\alpha,\beta)-E) \subset \cH_+$.
 Then the lowest eigenvalue  of $Q(\alpha,\beta)$ is simple.
 \end{theorem}

%To prove Theorem \ref{uniqthm}, prepare the next lemma.
%%%%%%%%%%%%%%%%
Let $Q_\qq =Q_{\qq1}\oplus Q_{\qq2}$, $\qq=+,-$. Then 
$Q$ is decomposed into the direct sum of even part and odd part, $Q=Q_+\oplus Q_-$.
Let $E_\qq=\inf Spec(Q_\qq)$. 
\begin{lemma}
Let $u=\vvv{u_1\\  u_2}$ be an eigenvector of $Q$. Then 
$u_j\in C^3(\RR)$ for $j=1,2$.
\end{lemma}
\proof
Let $u=\vvv{u_1\\  u_2}$ be an eigenvector of $Q$ with eigenvalue $\lambda$:
\begin{align}
\label{ev}
 A \left(-\frac{1}{2}\frac{d^2}{dx^2}+\frac{1}{2}x^2 \right) u
      + J \left(x\frac{d}{dx}+\frac{1}{2}\right)  u = \lambda u.
\end{align}
From the eigenvalue equation \kak{ev} we can directly see that 
\begin{align}
& \frac{1}{4} \frac{d^4}{dx^4} u 
 =
 \left( -\frac{1}{2}\frac{d^2}{dx^2} \right)
  \left\{ \lambda A^{-1} -\frac{x^2}{2}-A^{-1}J\left(x \frac{d}{dx}+ \frac{1}{2} \right) \right\} u \nonumber \\
%& = \frac{1}{2}\left[ \frac{d^2}{dx^2}, \frac{x^2}{2} + A^{-1}J x \frac{d}{dx} \right]
   %  u + \left\{ \lambda  A^{-1} -\frac{x^2}{2}-A^{-1}J\left(x \frac{d}{dx}+\frac{1}{2}\right) \right\}^2 
%    u \\
& \label{sobolev}
= \left(x\frac{d}{dx}+\frac{1}{2} \right) u+A^{-1}J \frac{d^2}{dx^2} u
     + \left\{ \lambda  A^{-1} -\frac{x^2}{2}-A^{-1}J\left(x \frac{d}{dx}+\frac{1}{2}\right) \right\}^2 u
\end{align}
in the sense of distribution. 
Note that $u_j\in D(x^2)\cap D(d^2/dx^2)$. 
Since $x^2 u_j, d^2 u_j/dx^2\in \LR$, 
we see that $ u_j \in W_\mathrm{loc}^{4,2}(\BR)$ for $j=1,2$ by \kak{sobolev}.
By the Sobolev embedding theorem, $ u_1, u_2 \in C^3(\BR)$ follows. 
\qed

\begin{lemma}{\label{lem4.2}} 
It follows that $E_+ \leq E_-$.
\end{lemma}
%%%%%%%%%%%%%%%%
 \begin{proof}
 Let $\Phi_- = \vvv{\Phi_{-1}\\ \Phi_{-2}}$ be a
normalized ground state of $Q_-$. 
Note that $\Phi_{-j }$, $j =1,2$,  are odd functions.
We define  even functions $\UU _- \in \cH_+$ by 
\begin{align*}
&  \UU _- =  \vvv {\UU  _{-1}\\  \UU  _{-2}}, \quad 
\UU  _{-j}(x) = \begin{cases}
			 \Phi_{-j }(x), \quad \text{if} ~ x\geq 0, \\
			-\Phi_{-j }(x), \quad \text{if} ~ x <  0. \\
			\end{cases}
\end{align*}
Note that $\UU _{-j } \in D(-d^2/dx^2)$ and 
\begin{align*}
& \norm{(d/dx)\UU _{-j}}^2 = \norm{(d/dx)\Phi_{-j}}^2,  \quad 
 \inner{\UU  _{-j '}}{ x \frac{d}{dx} \UU  _{-j }} = 
  \inner{ \UU_{-j '} }{ x \frac{d}{dx}  \UU_{-j } },
 \quad j',j  = 1,2.
 \end{align*}
%see \cite[Theorem 6.17]{Lieb-Loss:2001}.
Thus one has
\begin{align}
\label{58a}
E_+ \leq  \inner{\UU _-}{Q\UU _-} = \inner{\Phi_-}{Q\Phi_-}= E_-.
\end{align}
Therefore $E_+\leq E_-$ follows. 
\qed

\begin{lemma}
It follows that $E_+<E_-$.
\end{lemma}
\proof
Assume that $E_+=E_-$. 
Then by \eqref{58a}  we have
$E_+ = \inner{\UU _-}{Q\UU _-}$,
 which implies that 
$\UU  _-$ is a ground state of $Q_+$.
In other words, $\UU _-$ is an eigenvector of $Q$ with eigenvalue $E_+$.
Thus $\UU_{-j},  \UU  _{-j} \in C^3(\BR) $ for $j=1,2$. 
We normalize $\UU$ as $\|\UU\|=1$. 
From the fact that $\Phi_{-j}$ is odd (resp. $\UU  _{-j}$ is even), 
it follows that 
$\Phi_-(0)=\vvv {0\\ 0}=\UU_-(0)$ (resp. $\frac{d}{dx} \UU  _{-j}(0)=\vvv{0\\0}$).
Therefore $\UU_{-j}$ satisfies the  ordinary differential equations:
\begin{align}
 \label{ode512}
&\frac{d}{dx}
\vvv{
 \UU_{-1} \\ \UU_{-2} \\ \UU'_{-1} \\ \UU  '_{-2}
}
=
\left(
\!\!\!
\begin{array}{cccc}
0&0&1&0\\
0&0&0&1\\
x^2+\frac{2E_+}{\alpha}&-\frac{1}{\alpha}&0&-\frac{2x}{\alpha}\\
\frac{1}{\beta}& x^2-\frac{2E_+}{\beta}&\frac{2x}{\beta}& 0
\end{array}
\!\!\!
\right)
\vvv{
 \UU_{-1} \\ \UU_{-2} \\ \UU'_{-1} \\ \UU  '_{-2}
}
\\
&\label{ode513}
\vvv{
 \UU  _{-1} (0)\\ \UU  _{-2}(0) \\ \UU  '_{-1} (0)\\ \UU  '_{-2}(0)
}
=
\vvv{
0\\ 0 \\ 0 \\ 0
}.
\end{align}
Since the right hand side of  \kak{ode512} is smooth in $(\UU  _{-1},\UU  _{-2}, \UU  '_{-1}, \UU  '_{-2},x)$, the differential equation \eqref{ode512} with 
initial condition \kak{ode513} 
has the  unique solution
$\vvv{
 \UU  _{-1} (x)\\ \UU  _{-2}(x) \\ \UU  '_{-1} (x)\\ \UU  '_{-2}(x)
}
=
\vvv{
0\\ 0 \\ 0 \\ 0
}$, 
which contradicts $\norm{\UU  _-}=1$.
Therefore, $E_+ < E_-$.
 \end{proof}
\begin{proof}[Proof of Theorem \ref{uniqthm}] 
Assume that $\alpha\neq \beta$.
By Proposition \ref{waka}, it is enough to show 
that 
$ \ker(Q-E)\subset \cH_+$.
By Lemma \ref{lem4.2}, we have $E_+ < E_-$. 
Hence all the ground states are  even.
%The multiplicity of the lowest eigenvalue  
%of $Q_+ = Q_{+1} \oplus Q_{+2}$ is at most two, 
%since each of $Q_{+1}$ and $Q_{+2}$ has the simple lowest eigenvalue  by Theorem \ref{42}. 
Therefore the theorem follows.
%the lowest eigenvalue of  $Q$ is simple by Proposition \ref{waka}.
 \end{proof}

\section{Positivity of ground state}
%improving semigroup and positive ground state}
{\label{sec.pio}}
%In this section, we prove the semigroup $e^{-t\bar{Q}_{\qq\pp}}, \qq=\pm,\pp=1,2$ are positivity improving 
%semigroup. As a consequence, the ground states of $\bar{Q}_{\qq,\pp}, \qq=\pm, \pp=1,2$ are strictly positive.
Let 
\begin{align*}
&\con=\left\{\left.
\sum_{n=0}^\infty a_n \vvv{\ket{4n}\\0}
 + \sum_{n=0}^\infty b_n 
\vvv{0\\ \ket{4n+2}}
\right|a_n>0, b_n>0, n\geq 0\right\},\\
&
\conn=\left\{\left.
\sum_{n=0}^\infty a_n \vvv{\ket{4n}\\ 0} + \sum_{n=0}^\infty b_n \vvv{ 0\\ \ket{4n+2}}\right|a_n\geq 0, b_n\geq 0, n\geq 0\right\}.
\end{align*}
Then $\con$ is  a positive cone of $\hhh_{+1}$ and 
$\conn$  a non-negative 
 cone of $\hhh_{+1}$. 
%Thus $\RR_+ \con \subset \con$ and $\RR_+^0\conn \subset \conn$.
We say that $\Psi$ is non-negative and  is denoted by  $\Psi \geq 0$ 
if and only if 
$\Psi\in\conn$, 
and 
is strictly positive and is denoted by $\Psi>0$ if and only if  $\Psi\in\con$.
%Similarly, we define the positiveness for vectors in $\Ran(T_{+2})$, $\Ran(T_{-1})$
%and $\Ran(T_{-2})$.
%%%%%%%%%%%%%%%%
A bounded operator $T$ on $\hhh_{+1}$ is positivity preserving if and only if 
$T\conn\subset \conn$, and positivity improving if and only if $T\conn\subset\con$. 
\begin{proposition}
\label{faris}
Suppose that a bounded self-adjoint operator $T$ is positivity improving on $\hhh_{\qq\pp}$ and 
$\|T\|$ is an eigenvalue. Then the multiplicity of $\|T\|$ is simple and the corresponding
eigenvector is strictly positive.
\end{proposition}
\begin{proof}
See \cite{Faris:1972}.
\end{proof}

\begin{theorem}{\label{31}}
For all $t>0$,  $\qq=\pm$ and $\pp=1,2$, 
$e^{-t\til{Q}_{\qq\pp}}$
is positivity improving on $\hhh_{\qq\pp}$. 
In particular, the lowest eigenvalue  of $Q_{\qq\pp}$ is simple and corresponding eigenvector
is strictly positive.
\end{theorem}
 \begin{proof}
We prove the theorem only for the case of $\qq=+$ and $\pp=1$.
For  other cases the proof is similar and is left to readers.  
We  shall below show that $e^{-t\til{Q}_{+1}}$ is positivity improving.
We define
\begin{align}
  H_0 = A (a^* a + \frac{1}{2}) T_{+1} , \qquad 
  V =    \frac{\OO }{2}(aa+\add\add) T_{+1}.
\end{align}
Note that $\til{Q}_{+1} = H_0-V$.
Since $a\ket{n}= \sqrt{n}\ket{n-1}$ and $a^* \ket{n} = \sqrt{n+1}\ket{n+1}$, 
and 
$H_0$ is the multiplication by  
$\alpha(n+\frac{1}{2})$,
we see that 
$e^{-tH_0}$ is  positivity preserving.
Since  $\vvv{\ket{4n}\\ 0}$ and $\vvv{0\\ \ket{4n+2}}$ are 
analytic vectors of $V$, 
we see that 
\begin{align}
  e^{tV}\vvv{\ket{4n}\\ 0} &= \sum_{j=0}^\infty \frac{t^j}{j!} 
  (aa+\add\add)^j \left(\frac{\OO}{2}\right)^j \vvv{\ket{4n}\\ 0} \in\con , \\
  e^{tV}\vvv{0\\ \ket{4n+2}} &= \sum_{j=0}^\infty \frac{t^j}{j!} 
  (aa+\add\add) ^j \left(\frac{\OO}{2}\right)^j \vvv{0\\ \ket{4n+2}} \in\con.
\end{align}
From this 
$e^{tV}\conn\subset\con$ follows.  
Let $\Psi,\Phi \in \conn$.
By  the Trotter-Kato product formula, we have
\begin{align}
 \inner{\Psi}{e^{-t\til{Q}_{+1}}\Phi} = \lim_{j\to \infty}\inner{\Psi}{ (e^{-tH_0/j} e^{tV/j})^j\Phi} \geq 0.
\end{align}
Therefore $e^{t \til{Q}_{+1}}$ is positivity preserving. 
Next we  show that $e^{-t\til{Q}_{+1}}$ is positivity improving.  
We can assume that $\alpha \leq \beta$ without loss of generality.
Let $P_{\leq k}$ be the projection defined by
$$
  P_{\leq k} = 
\begin{pmatrix}
  \sum_{4n\leq k} \ket{2n}\bra{4n}& 0\\
0&   \sum_{4n+2\leq k} \ket{4n+2}\bra{4n+2}
\end{pmatrix}
$$
It is immediately seen that 
$\Psi \geq P_{\leq k} \Psi $  for any $\Psi \in \conn$ 
and 
 $e^{tV/j}\Psi \geq (1+tV/j)\Psi$.
For $k'\geq k$, we set $v = \vvv {\ket{4k}\\0}$ and $v'=\vvv {0\\ \ket{4k'}}$.
Then we have
\begin{align*}
  \inner{v'}{e^{-t\til{Q}_{+1}} v} &= \lim_{j\to \infty} \inner{v'}{ (e^{-tH_0/j} e^{tV/j})^j v}  
%  & \geq  \lim_{j\to \infty}\inner{v'}{ (e^{-tH_0/j} e^{tV/j})^{j-1} (e^{-tH_0/j} P_{\leq k'} e^{tV/j}) v} \\
  \geq  \varlimsup_{j\to \infty}\inner{v'}{ (e^{-tH_0/j} P_{\leq k'}  e^{tV/j})^j v}  \\
  & \geq  \varlimsup_{j\to \infty}\inner{v'}{ (e^{-t(k'+(1/2))\beta/j} P_{\leq k'}e^{tV/j})^j v} \\
&   \geq  e^{-t(k'+(1/2))\beta} \varlimsup_{j\to \infty}\inner{v'}{ (P_{\leq k'} (1+tV/j))^j v}.
\end{align*}
Note that $e^{-tH_0}(1+tV/j)$ is still positivity preserving.
For all $\ell=2k'-2k$, we have
\begin{align*}
&  \varlimsup_{j\to \infty}\inner{v'}{ ( P_{\leq k'} (1+tV/j))^j v}
 \geq 
   \varlimsup_{j\to \infty}\inner{v'}{ {}_jC_\ell (tV/j)^\ell v} 
 \geq
   \varlimsup_{j\to \infty}\inner{v'}{ {}_jC_\ell (t(a^*)^2/j)^\ell v} \\
& =
  t^\ell  \varlimsup_{j\to \infty} {}_jC_\ell j^{-\ell} \inner{v'}{ (a^*)^{4k'-4k} v} 
 =
   \frac{t^\ell}{\ell !}  \varlimsup_{j\to \infty} \frac{j(j-1)\cdots (j-\ell-1)}{j^\ell } \inner{v'}{ (a^*)^{4k'-4k} v} \\
& =
   \frac{t^\ell}{\ell !}  \inner{v'}{ (a^*)^{4k'-4k} v} >0,
\end{align*}
where $_jC_k$ denotes the binomial coefficient.
Thus we have $\inner{v'}{e^{-t\til{Q}_{+1}}v}>0$.
Similarly  
$
  \inner{\vvv{\ket{4n}\\ 0}}{e^{-t\til{Q}_{+1}}\vvv{0\\ \ket{4n+2}}}>0$ is derived 
  for all $n$.
Thus $e^{-t\til{Q}_{+1}}$ is positivity improving. 
 \end{proof}

\section{No  crossings}
\label{sec.difference}
Recall that $E_n(\alpha,\beta)$ be  the $n$-th eigenvalue of $Q(\alpha,\beta)$, 
and the map $(\alpha,\beta)\mapsto E_n(\alpha,\beta)\in\RR$ is 
an  eigenvalue-curve.
It is shown that the  spectrum of $Q$ 
is $Spec (Q)=\bigcup_{\qq=\pm,\pp=1,2} Spec (Q_{\pp\qq})$, 
and 
all the eigenvalues in $Spec(Q_{\pp\qq})$  are simple. 
%As is shown in Section \ref{sec.jacobi}, 
%operator $Q_{\qq\pp}$
%can be represented as the Jacobi matrix $\widehat{Q}_{\qq\pp}$.
Now we are interested in operators,   
$\widehat{Q}_{-1}- \widehat{Q}_{+1}$ and 
$\widehat{Q}_{-2} - \widehat{Q}_{+2}$.
%%%%%%%%%%%%%%%%
\begin{theorem}\label{5.1}
Assume that 
\begin{align}
 \sqrt{ \alpha\beta} >  1 + \frac{1}{1600000000} 
\end{align}
Then
$\widehat{Q}_{-1} - \widehat{Q}_{+1} \geq \Delta(\alpha,\beta)$ and 
 $ \widehat{Q}_{-2} - \widehat{Q}_{+2} \geq   \Delta(\alpha,\beta)$, 
where 
\[
\Delta(\alpha,\beta) = 2\min\{\sqrt{\alpha/\beta}, \sqrt{\beta/\alpha}\} (\sqrt{\alpha\beta}-1-1/1600000000)> 0. 
\]
In particular $\lambda_{-1}(n)\geq  \lambda_{+1}(n)+\Delta(\alpha,\beta)$ and 
$\lambda_{-2}(n)\geq  \lambda_{+2}(n)+\Delta(\alpha,\beta)$.
\end{theorem}
%%%%%%%%%%%%%%%%
\begin{proof}
We have 
\begin{align}
  \widehat{Q}_{-1}- \widehat{Q}_{+1} 
 = 
{\tiny  \frac{1}{2}
\left(\begin{array}{rrrrrrrr}
2\alpha &  - \gamma_0 &  &  &  &  & & \text{\huge 0} \\
-\gamma_0 & 2\beta & -\gamma_1 &  &  &  &  & \\
 & -\gamma_1 & 2\alpha & -\gamma_2 &  &  &  &\\
 &  & -\gamma_2 & 2\beta & -\gamma_3 &  &  & \\
 &  &  & -\gamma_3 & 2\alpha & -\gamma_4 &  & \\
 &  &  &  & -\gamma_4 & 2\beta & -\gamma_5 &\\
 &  &  &  &  & -\gamma_5 & 2\alpha & \ddots \\
\text{\huge 0} & & &&&&\ddots& \ddots
\end{array}\right),
}\end{align}
where $\gamma_n = \sqrt{(2n+2)(2n+3)}-\sqrt{(2n+1)(2n+2)}$.
We set
\begin{align}
  S_1 = \diag[(\beta/\alpha)^{1/4}, (\alpha/\beta)^{1/4}, (\beta/\alpha)^{1/4}, (\alpha/\beta)^{1/4}, \cdots], \\
  S_2 = \diag[(\alpha/\beta)^{1/4}, (\beta/\alpha)^{1/4}, (\alpha/\beta)^{1/4}, (\beta/\alpha)^{1/4}, \cdots].
\end{align}
Then we have
\begin{align}
& S_1 (\widehat{Q}_{-1} - \widehat{Q}_{+1}) S_1=
 S_2 (\widehat{Q}_{-2} - \widehat{Q}_{+2}) S_2= \\
&= {\tiny   \frac{1}{2}
\left(\begin{array}{rrrrrrrr}
2\sqrt{\alpha\beta} &  - \gamma_0 &  &  &  &  & \text{\huge 0}  \\
-\gamma_0 & 2\sqrt{\alpha\beta}& -\gamma_1 &  &  &  & \\
 & -\gamma_1 & 2\sqrt{\alpha\beta} & -\gamma_2 &  &  & \\
 &  & -\gamma_2 & 2\sqrt{\alpha\beta}& -\gamma_3 &  & \\
 &  &  & -\gamma_3 & 2\sqrt{\alpha\beta} & -\gamma_4 &  \\
 &  &  &  & -\gamma_4 & 2\sqrt{\alpha\beta}& \ddots \\
\text{\huge 0} & & &&&\ddots & \ddots
\end{array}\right).
}\end{align}
We set $ F = J((\gamma_n)_{n=0}^\infty, 0)$. 
Then $S_1(\widehat{Q}_{-1} - \widehat{Q}_{+1})S_1 = 2\sqrt{\alpha\beta}-F$.
Since $S_1$ is self-adjoint and invertible, we have
\begin{align*}
 & (\widehat{Q}_{-1} - \widehat{Q}_{+1}) \geq  (2\sqrt{\alpha\beta}-\norm{F})S_1^{-2} \geq (2\sqrt{\alpha\beta}-\norm{F}) \min\{\sqrt{\alpha/\beta},\sqrt{\beta/\alpha}\}.
\end{align*}
Similarly we have $(\widehat{Q}_{-2} - \widehat{Q}_{+2}) \geq (2\sqrt{\alpha\beta}-\norm{F}) \min\{\sqrt{\alpha/\beta},\sqrt{\beta/\alpha}\}$.
Hence it is sufficient to prove $\norm{F}<2(1+1/1600000000)$.
Let $v = (v_n)_{n=0}^\infty \in \ell^2$.
Then we have
\begin{align*}
 |\inner{v}{Fv}| &= 
  \left| \sum_{n=0}^\infty (\overline{v_n} \gamma_n v_{n+1} +{v_n} \gamma_n \overline{v_{n+1}} ) \right|   \\
      &\leq 2 \sum_{n=0}^\infty |v_n| \gamma_n |v_{n+1}|                          \leq  \sum_{n=0}^\infty \left( a_n |v_n|^2 + \frac{\gamma_{n}^2}{a_n} |v_{n+1}|^2 \right)
\end{align*}
for any $a_n > 0$. So it follows that 
\begin{align}
 |\inner{v}{Fv}| &\leq a_0 |v_0|^2 + \sum_{n=1}^\infty (a_n + \frac{\gamma_{n-1}^2}{a_{n-1}}) | v_n |^2.  \label{515}
\end{align}
We split \eqref{515} as 
\begin{align}
|\inner{v}{Fv}|  & \leq  a_0 |v_0|^2 + \sum_{n=1}^{N_0} (a_n + \frac{\gamma_{n-1}^2}{a_{n-1}}) | v_n |^2 
                         + \sum_{n=N_0+1}^\infty (a_n + \frac{\gamma_{n-1}^2}{a_{n-1}}) | v_n |^2 \label{bound.k}
\end{align}
with some $N_0$. 
We recursively define $a_n$  by 
\begin{align}
   a_0 = 2, 
\quad     a_n = 2 - \frac{\gamma_{n-1}^2}{a_{n-1}} \ ( n=1,2,3, \cdots, N_0) , 
\quad     a_n = 1 \ ( n \geq N_0+1). \label{def.an}
\end{align}
We  can compute the numerical value of $a_n$ from \eqref{def.an}, e.g. 
$  a_1 = 1.464\cdots$, $a_2 = 1.305\cdots$, $a_3 = 1.228\cdots$.
We take $N_0=10000$. Then one can easily check that 
$a_n>0$ for all $n < N_0$ and $a_{N_0}>1$.
Hence the inequality \eqref{bound.k} is valid for $N_0=10000$ and we have
\begin{align*}
|\inner{v}{Fv}|  & \leq  2 |v_0|^2 + 2 \sum_{n=1}^{N_0} | v_n |^2 
                         + \sum_{n=N_0+1}^\infty (a_n + \frac{\gamma_{n-1}^2}{a_{n-1}}) | v_n |^2 \\
& <  2 |v_0|^2 + 2 \sum_{n=1}^{N_0} | v_n |^2 
                         + \sum_{n=N_0+1}^\infty (1 + \gamma_{n-1}^2) | v_n |^2.
\end{align*}
where we used \eqref{def.an}.
On the other hand, we have
$
 \gamma_{n-1}^2 
= 1 + \frac{1}{(2n+\sqrt{4n^2-1})^2} $.
In particular $\gamma_{n-1}$ is monotonously decreasing. Therefore we have
\begin{align}
 |\inner{v}{Fv}|  & \leq  (1 + \gamma_{N_0}^2) \sum_{n=0}^\infty  | v_n |^2,
\end{align}
which implies that $\norm{F} \leq 1 + \gamma_{N_0}^2$.
Note that 
 \begin{align}
   \gamma_{N_0}^2 < \gamma_{N_0-1}^2 < 1 + \frac{1}{(4N_0)^2}= 1 + \frac{1}{1600000000}.
 \end{align}
Therefore $\norm{F}< 2(1+1/1600000000)$.
\end{proof}

The map 
$(\alpha,\beta)\mapsto \lambda_{\qq\pp}(n)=\lambda_{\qq\pp}(n, \alpha,\beta)$ 
is an eigenvalue-curve. 
It is immediate to see the corollary below by Theorem \ref{5.1}. . 
\begin{corollary}\label{nocrossing}
Let 
$$D=\{(\alpha,\beta)\in\RR\times\RR| \alpha>0,\beta>0, \alpha\ne\beta,
 \sqrt{ \alpha\beta} >  1 + \frac{1}{1600000000} 
\}.$$
Fix $\pp=1,2$. 
Then two eigenvalue-curves $\lambda_{-\pp}(n)$ and $\lambda_{+\pp}(n)$ have 
no crossing 
in the region $D$ for all $n$.
\end{corollary}

\section{Numerical results}{\label{sec.finite}}
For finite sequences $a=(a_0,\cdots,a_{N-1})$ and $b=(b_0,\cdots,b_N)$, 
we define  
the $(N+1)$-dimensional Jacobi matrix, $J(a,b)$,  by
\begin{align}
 J(a,b) = 
\begin{pmatrix}
  b_0 & a_0 &       &            & \text{\Large 0}\\
  a_0 & b_1 & a_1 &    \\
        &  a_1 & b_2 &  \ddots & \\
        &        &  \ddots & \ddots & a_{N-1} \\
\text{\Large 0}       & &        & a_{N-1} & b_N
\end{pmatrix}.
\end{align}
For $\qq=\pm$ and $\pp=1,2$, we set
$a_\qq^N=(a_\qq(n))_{n=0}^{N-1}$ and 
$b_{\qq\pp}^N= (b_{\qq\pp}(n))_{n=0}^N $. 
Define a  finite Jacobi matrix by 
$
 \widehat{Q}_{\qq\pp}(N) = \frac{1}{2} J( a_\qq^N, b_{\qq\pp}^N)$.
We set 
\begin{align}
 & \Lambda_{+1}(N) =  \frac{1}{2} (\alpha\beta-1)\times
 \begin{cases}
    \min\{ \alpha^{-1}(2N+\frac{3}{2}), \beta^{-1}(2N+\frac{7}{2})  \} \quad 
    \text{if } N \text{ is even}\\
    \min\{ \beta^{-1}(2N+\frac{3}{2}), \alpha^{-1}(2N+\frac{7}{2})  \} \quad 
    \text{if } N \text{ is odd} 
 \end{cases}\\
 &\Lambda_{+2}(N) = \Lambda_{+1}(N)\Big|_{(\alpha,\beta)\to(\beta,\alpha)}  \\
 & \Lambda_{-1}(N) = \frac{1}{2}   (\alpha\beta-1)
 \begin{cases}
 \min\{ \alpha^{-1} (2N+\frac{5}{2}), \beta^{-1} (2N+\frac{9}{2})  \} \quad 
   \text{if } N \text{ is even}\\
  \min\{ \beta^{-1} (2N+\frac{5}{2}), \alpha^{-1} (2N+\frac{9}{2})  \} \quad 
   \text{if } N \text{ is odd} 
 \end{cases}\\
& \Lambda_{-2}(N) = \Lambda_{-1}(N) \Big|_{(\alpha,\beta)\to(\beta,\alpha)}
\end{align}
and
\begin{align}
\delta_{\pm,1}(N) = \begin{cases}
		\frac{1}{2}\alpha |a_\pm (N) | \quad \text{ if $N$ is even}\\
		\frac{1}{2}\beta  |a_\pm (N) | \quad  \text{ if $N$ is odd},
	       \end{cases}\quad 
\delta_{\pm,2}(N) = \delta_{\pm,1}(N) \Big|_{(\alpha,\beta)\to(\beta,\alpha)}. 
\end{align}
Since $\alpha\beta>1$, one has
$
 \Lambda_{\qq\pp}(N) = O(N) \to +\infty ~ (N\to +\infty)$.
Let $p_n$ be the orthogonal projection onto 
$e_n = (\delta_{n,j})_{j=0}^\infty \in \ell^2$. 
For a self-adjoint operator $T$, $\mu_n(T)$, $n=1,2,\cdots$, 
 denotes the $n$-th eigenvalue of $T$ counting multiplicity.  
For $n=0,1,\cdots, N$, we set 
\begin{align*}
 &   \lambda_{\qq\pp,N}(n) = \mu_n(\widehat{Q}_{\qq\pp}(N)), \\
 &   \lambda_{\qq\pp,N}^{\mathrm{upper}}(n) = 
        \mu_n( \widehat{Q}_{\qq\pp}(N) + \delta_{\qq\pp}(N) p_N), \\
 &    \lambda_{\qq\pp,N}^{\mathrm{lower}}(n) = 
    \mu_n( \widehat{Q}_{\qq\pp}(N) - \delta_{\qq\pp}(N) p_N).
\end{align*}
The eigenvalues of $\widehat{Q}_{\qq\pp}$ can be approximated by the eigenvalues of 
the $(N+1)$-dimensional matrix $\widehat{Q}_{\qq\pp}(N)$ in the following sense.
%%%%%%%%
 \begin{theorem}{\label{thm6.1}}
Fix $N\in \BN$, $\qq=\pm$ and $\pp=1,2$. 
Let $n\in \BN$ be a number such that
\begin{align}
  \lambda_{\qq\pp,N}^{\mathrm{upper}}(n)  \leq  \Lambda_{\qq\pp}(N). 
  \label{6.8oi}
\end{align}
Then it follows that 
\begin{align}
\lambda_{\qq\pp,N}^{\mathrm{lower}}(n) 
 \leq \lambda_{\qq\pp}(n)  \leq \lambda_{\qq\pp,N}^{\mathrm{upper}}(n) 
\label{610ye}
\end{align}
In particular, the error is estimated as 
$
  |  \lambda_{\qq\pp}(n) -  \lambda_{\qq\pp,N}(n) )  |
\leq
\lambda_{\qq\pp,N}^\mathrm{upper}(n) - \lambda_{\qq\pp,N}^\mathrm{lower}(n) $.
 \end{theorem}

We give an example below:
%An application of Theorem \ref{thm6.1} is discussed in Example \ref{ex6.2} below.
%%%%%%%%%%%%%%%%
 \begin{example}{\label{ex6.2}}
We set $\mathcal{Q}_\pm=\widehat{Q}_{+1}(N) \pm \delta_{+1}(N)p_N$.
We apply Theorem \ref{thm6.1} to the case
$\alpha=1$,
  $\beta=2$ and 
  $N=10$. 
  Then $\Lambda_{+1}(N)=5.875$ and 
$$
\begin{array}{ll}
\lambda_{+1,N}^\mathrm{upper}(0) = 0.366917859 \pm 0.000000001, \quad 
& \lambda_{+1,N}^\mathrm{lower}(0)  = 0.366917862 \pm 0.000000001, \\
 \lambda_{+1,N}^\mathrm{upper}(1) = 2.432911 \pm 0.000001, \quad
& \lambda_{+1,N}^\mathrm{lower}(1)  = 2.432920 \pm 0.000001, \\ 
\lambda_{+1,N}^\mathrm{upper}(2) = 4.7145 \pm 0.0001, \quad
& \lambda_{+1,N}^\mathrm{lower}(2)  = 4.7164 \pm 0.0001\\
 \lambda_{+1,N}^\mathrm{upper}(3) = 6.2717 \pm 0.0001, \quad
 & \lambda_{+1,N}^\mathrm{lower}(3)  = 6.2789 \pm 0.0001.
\end{array}
$$
Since $\lambda_{+1,N}^\mathrm{upper}(2) \leq \Lambda_{+1}(N)=5.875$, by Theorem \ref{thm6.1} we have 
 numerical bounds:
\begin{align*}
 0.36691785 \leq  &\lambda_{\qq\pp}(0) \leq 0.36691786,  \\
 2.43291  \leq       &\lambda_{\qq\pp}(1) \leq  2.43292, \\
  4.714    \leq        &\lambda_{\qq\pp}(2) \leq  4.717.
\end{align*}
This example does not include  the bound on
 $\lambda_{\qq\pp}(3)$, since the condition \eqref{6.8oi} is not valid for $n=3$.
 \end{example}
%%%%%%%%%%%%%%%%

 \begin{proof}[Proof of Theorem \ref{thm6.1}:]
We prove the theorem only for the case of $\qq=+$ and $\pp=1$. 
The other cases can be  similarly proven. 
For $u,v\in \ell^2$, we define the operator $u \odot v:\ell^2\to\ell^2$ by
$
     (u \odot v ) \Phi = \inner{v}{\Phi} u,$ for $\Phi \in \ell^2$.
Then operator $\widehat{Q}_{+1}$ can be expressed as
\begin{align*}
 \widehat{Q}_{+1} 
=
 \widehat{Q}_{+1}(N)\oplus 0
 + \sum_{n=N+1}^\infty b_{+1}(n) p_n
 + \sum_{n=N}^\infty a_+(n)  (e_n \odot e_{n+1} + e_{n+1} \odot e_n).
\end{align*}
We  can  show that $u\odot v +v\odot u \leq \ep u\odot u + \ep^{-1} v\odot v$ 
for all $\ep>0$. 
By using this inequality, we have
\begin{align*}
&  \sum_{n=N}^\infty a_+(n)  (e_n \odot e_{n+1} + e_{n+1} \odot e_n)
\leq 
  \sum_{n=N}^\infty  |a_+(n)|  (\ep_n e_n \odot e_n + \ep_n^{-1} e_{n+1} \odot e_{n+1}) \\
& = |a_+(N)| \ep_N p_N
 + \sum_{n=N+1}^\infty  (\ep_n |a_+(n)|+\ep_{n-1}^{-1}|a_+(n-1)|)  p_n
\end{align*}
for all $\ep_n>0$.
We take $\ep_{2n+1}=\beta$ and  $\ep_{2n}=\alpha$ for even $N$, and 
 $\ep_{2n+1}=\alpha$ and $\ep_{2n}=\beta$ for odd $N$.
Note that $|a_+(N)| \ep_N = \delta_{+1}(N) $.
First we consider the case of even  $N$.
Then, we have
\begin{align*}
&  \sum_{n=N+1}^\infty  (\ep_n |a_+(n)|+\ep_{n-1}^{-1}|a_+(n-1)|)  p_n \\
& =\sum_{n=0}^\infty  (\ep_{N+n+1} |a_+(N+n+1)|+\ep_{N+n}^{-1}|a_+(N+n)|)  
    p_{N+n+1} \\
& = \sum_{n=0}^\infty  (\ep_{N+2n+1} |a_+(N+2n+1)|+\ep_{N+2n}^{-1}|a_+(N+2n)|)  
    p_{N+2n+1} \\
&\quad + \sum_{n=0}^\infty  (\ep_{N+2n+2} |a_+(N+2n+2)|+\ep_{N+2n+1}^{-1}|a_+(N+2n+1)|)  
    p_{N+2n+2} \\
& = \sum_{n=0}^\infty  (\beta |a_+(N+2n+1)|+\alpha^{-1}|a_+(N+2n)|)  
    p_{N+2n+1} \\
& \quad + \sum_{n=0}^\infty  (\alpha |a_+(N+2n+2)|+\beta^{-1}|a_+(N+2n+1)|)  
    p_{N+2n+2}.
\end{align*}
Since $|a_+(n)|\leq 2n+\frac{3}{2}$, we have
\begin{align*}
&  \sum_{n=N+1}^\infty  (\ep_n |a_+(n)|+\ep_{n-1}^{-1}|a_+(n-1)|)  p_n \\
& \leq \sum_{n=0}^\infty  (\beta (2N+4n+2+\frac{3}{2})+\alpha^{-1}(2N+4n+\frac{3}{2}))  
       p_{N+2n+1} \\
& \quad + \sum_{n=0}^\infty  (\alpha (2N+4n+4+\frac{3}{2})+\beta^{-1}(2N+4n+2+\frac{3}{2}))  
       p_{N+2n+2}.
\end{align*}
By the definition of $b_{+1}(n)$, we have
\begin{align*}
 \widehat{Q}_{+1} 
\geq & ~
 \widehat{Q}_{+1}(N)\oplus 0  - \delta_{+1}(N) p_N\\
& + \frac{1}{2}\sum_{n=0}^\infty  \left(
         \beta(4N+8n+5) - \beta (2N+4n+\frac{7}{2}) - \alpha^{-1}(2N+4n+\frac{3}{2}) \right)
    p_{N+2n+1} \\
& + \frac{1}{2} \sum_{n=0}^\infty  \left(
     \alpha (4N+8n+9) - \alpha (2N+4n+\frac{11}{2})  -  \beta^{-1}(2N+4n+\frac{7}{2}) \right)  
    p_{N+2n+2} \\
 \geq & ~ \widehat{Q}_{+1}(N)\oplus 0  - \delta_{+1}(N) p_N
  + \frac{1}{2}\sum_{n=0}^\infty  (\beta-\alpha^{-1}) (2N+4n+\frac{3}{2}) 
    p_{N+2n+1} \\
& + \frac{1}{2} \sum_{n=0}^\infty (\alpha-\beta^{-1}) (2N+4n+\frac{7}{2}) 
    p_{N+2n+2}.
\end{align*}
Thus we have
$
 \widehat{Q}_{+1} \geq  (\widehat{Q}_{+1}(N) - \delta_{+1}(N) p_N) \oplus (\Lambda_{+1}(N))$.
We  can obtain  the same inequality for odd $N$.
In a similar way, we can furthermore obtain  the upper bound
$
 \widehat{Q}_{+1} \leq  (\widehat{Q}_{+1}(N) + \delta_{+1}(N) p_N)
  \oplus R(N)$,
where $R(N)$ is an operator such that $R(N)\geq \Lambda_{+1}(N)$.
By the min-max principle, we have
\begin{align*}
& \mu_n( (\widehat{Q}_{+1}(N) - \delta_{+1}(N) p_N)\oplus \Lambda_{+1}(N) )
 \leq  \mu_n(\widehat{Q}_{+1})   \\
&\leq 
 \mu_n( (\widehat{Q}_{+1}(N) + \delta_{+1}(N) p_N)\oplus R(N) ).
\end{align*}
Suppose that $\mu_n(\widehat{Q}_{+1}(N)+\delta_{+1}(N)p_N)\leq \Lambda_{+1}(N)$. 
Then
\begin{align*}
& \mu_n( (\widehat{Q}_{+1}(N) - \delta_{+1}(N) p_N)\oplus \Lambda_{+1}(N)) 
   = \mu_n(\widehat{Q}_{+1}(N)-\delta_{+1}(N) p_N),  \\
& \mu_n( (\widehat{Q}_{+1}(N)+\delta_{+1}(N) p_N)\oplus R(N) )
   =\mu_n(\widehat{Q}_{+1}(N)+\delta_{+1}(N) p_N).
\end{align*}
This proves \eqref{610ye}.
 \end{proof}

\section{Concluding remarks}
We can extend non-commutative harmonic oscillators to 
an infinite dimensional version. 
Let ${\mathscr F}=\oplus_{n=0}^\infty L_{\rm sym}^2({\mathbb R}^{n})$ be the boson Fock space, 
where $L_{\rm sym}^2({\mathbb R}^{n})$, $n\geq1$,  
denotes the set of symmetric square integrable functions, and $L^2(\RR^0)={\mathbb C}$.  
Let $a(f)$ and $\add (f)$, $f\in\LR$, be the annihilation operator and the creation operator, respectively, 
which satisfy  canonical commutation 
relations $[a(f), \add(g)]=(\bar f, g)$, $[a(f), a(g)]=0=[\add(f), \add(g)]$, and adjoint relation $(a(f))^\ast=\add(\bar f)$. 
Let $d\Gamma(\omega)=\int \omega(k)\add(k) a(k) dk$ be the second quantization of a real-valued multiplication 
$\omega$.
The scalar field is defined by $\phi(f)=\frac{1}{\sqrt 2}(\add(f)+a(\bar f))$ and 
its momentum conjugate by 
$\pi(f)=\frac{i}{\sqrt2}(\add(f)-a(\bar f))$. 
Thus we define the self-adjoint operator 
$$ H=A \otimes d\Gamma(\omega) + J \otimes 
\left(
i\phi(f) \pi(f)+\frac{1}{2} \|f\|^2
\right)$$
on ${\mathbb C}^2\otimes{\mathscr F}$.  
The spectrum of $H$ is not purely discrete.
It is interesting to consider  the existence of a ground state of $H$ and to estimate its multiplicity.

\section*{Acknowledgments}
IS thanks T. Nakamaru, H. Niikuni, T. Mine, S. Osawa,  A. Sakano and   H. Sasaki for useful comments.
IS's work was partly supported by Research supported by KAKENHI Y22740087, 
and was performed through the Program for Dissemination of Tenure-Track System  
funded by the Ministry of Education and Science, Japan.
FH acknowledges the financial support of Grant-in-Aid for Science Research (B) 20340032. 
We thank Masato Wakayama for his useful and helpful comments and also thank an invitation to the international conference 
"Spectral analysis of non-commutative harmonic oscillators and quantum devices" supported by 
the Ministry of Education, Culture, Sports, Science and  Technology in Japan and Institute of Math-for-Industry 
in Kyushu university 
in November of 2012.

\end{document}